\documentclass[11pt]{article}
\usepackage[utf8]{inputenc}

\usepackage{epsfig, amssymb, geometry}
\usepackage{latexsym, graphics, graphicx}
\usepackage{proof}

\usepackage{todonotes}

\usepackage{amsthm} 
\usepackage{amsmath}
\usepackage{amsxtra}
\usepackage{float}
\usepackage{authblk}
\usepackage{url}
\usepackage{mathptmx} 
\usepackage{enumerate}
\usepackage[numbers]{natbib}
\usepackage{fancyvrb}
\usepackage{marginnote}
\usepackage{mathtools}

\usepackage{todonotes}

\usepackage{blkarray}

\usepackage{tikz}
\usetikzlibrary{automata,positioning}

\usepackage{amssymb}
\newcommand{\ignore}[1]{}

\newcommand{\smalldownarrow}{\mbox{\tiny $\downarrow$}}

\newcommand{\X}{\mathcal{X}}

\newcommand{\Pos}{\mathcal{P}\!\mathit{os}}

\newcommand{\eq}{\mathcal{EQ}}

\newtheorem{defn}{Definition}

\newtheorem{appxlem}{Lemma}[section]

\newcommand{\xddots}{%
  \raise 4pt \hbox {.}
  \mkern 6mu
  \raise 1pt \hbox {.}
  \mkern 6mu
  \raise -2pt \hbox {.}
}

\newcommand*{\Scale}[2][4]{\scalebox{#1}{\ensuremath{#2}}}%

\DefineVerbatimEnvironment{todos}
 {Verbatim}
 {fontsize=\footnotesize}
\title{Asymmetric Unification and Disunification}
\author[1]{Veena Ravishankar} 
\author[2]{Kimberly A. Gero}% \\ gerok@strose.edu
\author[1]{Paliath Narendran}% \\ pnarendran@albany.edu
\affil[1]{University at Albany--SUNY}
\affil[2]{The College of Saint Rose}
\affil[ 1]{\textit {\{vravishankar,pnarendran\}@albany.edu}}
\affil[2 ]{\textit {gerok@strose.edu}}
\date{}%\today
\begin{document}
\maketitle

\begin{abstract}
We compare two kinds of unification problems: Asymmetric Unification
and Disunification, which are variants of Equational
Unification. Asymmetric Unification is a type of Equational
Unification where the right-hand sides of the equations are in normal
form with respect to the given term rewriting system. In
Disunification we solve equations and disequations with respect to an
equational theory for the case with free constants.
 We contrast the time complexities of both and show
that the two problems are incomparable: there are theories where one
can be solved in Polynomial time while the other is NP-hard. This goes
both ways. The time complexity also varies based on the termination
ordering used in the term rewriting system.
\end{abstract}

\section{Introduction and Motivation}
This is a short introductory survey on 
two variants of unification, namely asymmetric
unification~\cite{DBLP:conf/cade/ErbaturEKLLMMNSS13} and
disunification~\cite{BaaderSchulzDisunif95,DBLP:conf/birthday/Comon91}. 
We contrast the two in terms of their time complexities for
different equational
theories, for the case where terms in the input can
also have free constant symbols. Asymmetric unification is a new paradigm comparatively,
which requires one side of the equation to be
irreducible~\cite{DBLP:conf/cade/ErbaturEKLLMMNSS13}, while
disunification~\cite{DBLP:conf/birthday/Comon91} deals with
solving equations and disequations. Complexity analysis has been
performed separately on asymmetric
unification~\cite{DBLP:conf/rta/BrahmakshatriyaDGN13,DBLP:conf/fossacs/ErbaturKMMNR14}
and 
disunification~\cite{BaaderSchulzDisunif95,DBLP:journals/jacm/BuntineB94}, 
but not much work has been done on contrasting
the two paradigms.
In~\cite{DBLP:conf/cade/ErbaturEKLLMMNSS13}, it was shown that there are theories which are decidable for symmetric unification but are undecidable for asymmetric unification, so here we investigate this further.
Initially, it was thought that the two are
reducible to one another~\cite{DBLP:conf/fossacs/ErbaturKMMNR14}, but our results indicate
that they are not at least where time complexity is concerned.
In our last section we show that the time complexity of asymmetric unification varies depending on the symbol ordering chosen for the theory.\\
\noindent
Unification deals with solving symbolic equations. A solution to
a unification problem is a \emph{unifier}, a
substitution of certain variables by another
expression or term. Often we need to find \emph{most general unifiers (mgu)}.

For example, given two terms $u=f(a, y)$ and $v=f(x, b)$,
where $f$ is a binary function symbol, $a$ and $b$ are constants, and $x$ and $y$ are
variables, the substitution $\sigma = \{ x \mapsto a, ~ y \mapsto b\}$
unifies $u$ and~$v$.\\

\section{Notations and Preliminaries: Term Rewriting Systems, Equational Unification}

\noindent
We assume the reader is accustomed with the terminologies of term rewriting systems (TRS), equational rewriting~\cite{Term}, unification and equational unification~\cite{BaaderSnyd-01}.

\noindent
\underline{\bf Term Rewriting Systems}:
 A term rewriting system
(TRS)~\cite{Term} is a set of rewrite rules, where a rewrite rule is
an identity $l \, \approx \, r$ such that $l$ is not a variable and
$Var(l)$ $\supseteq$ $Var(r)$. It is often written or denoted as $l \rightarrow
r$. These oriented equations are commonly called {\em rewrite rules\/}.
%{\em Signature\/} $\Sigma$ refers to the set of function symbols occurring in~$R$.
The rewrite relation induced by~$R$ is written as
$\rightarrow_R^{}$.

A term is \emph{reducible} by a term rewriting system
if and only if a subterm of it is an instance of the left-hand side of
a rule. In other words, a term $t$ is {\em reducible\/} modulo~$R$ if and only if
there is a rule $l \rightarrow r$ in $R$, a subterm $t'$ at
position {\em p\/} of {\em t\/}, and a substitution $\sigma$ such that
$\sigma(l) = t'$. The term~$t[\sigma(r)]_p^{}$ is the result of
{\em reducing t\/} by $l \rightarrow r$ at {\em p}.  The {\em reduction
relation\/} $\rightarrow_{{R}}$ associated with a term rewriting
system~$R$ is defined as follows: $s \, \rightarrow_{{R}} \,
t$ if and only if there exist {\em p\/} in $Pos(s)$ and $l \rightarrow
r$ in~$R$ such that $t$ is the result of reducing $s$ by $l
\rightarrow r$ at~{\em p}, i.e., $t = s[\sigma(r)]_p^{}$.

A term is in \emph{normal form} with respect to a term rewriting
system if and only if no rule can be applied to it. A term
rewriting system is \emph{terminating} if and only if there are
no infinite rewrite chains.

Two terms $s$ and $t$ are said to be {\em joinable\/} modulo a term rewriting
system $R$ if and only if there exists a term~$u$ such that $s \;
\rightarrow_{R}^* \; u$ and $t \; \rightarrow_{R}^* \; u$, denoted as $s\downarrow t$.

The equational theory $\mathcal{E}(R)$ associated with a term rewriting
system~$R$ is the set of equations
obtained from~$R$ by treating every rule as a (bidirectional)
equation. Thus the equational congruence $\approx_{\mathcal{E}(R)}^{}$
is the congruence $({\rightarrow_{R}^{}} \cup {\leftarrow_{R}^{}})^*$.

A term rewriting system $R$ is said to be 
{\em confluent\/} if and only if the following (``diamond'') property holds:
 \[ \forall t \forall u \forall v\
 \left[ \vphantom{b^b} (t \rightarrow^{*}_{R} u \; \wedge \; t \rightarrow^{*}_{R} v)
~ \;  \Rightarrow \; ~
 \exists w (u \rightarrow^{*}_{R} w \; \wedge \; v \rightarrow^{*}_{R} w) \right]
\] $R$
is \emph{convergent} if and only if it is terminating and confluent.
In other words, $R$
is \emph{convergent} if and only if it is terminating and, besides, every term
has a \emph{unique} normal form. 
 
An equational theory $\approx_E^{}$ is said to be \emph{subterm-collapsing}
if and only if there exist terms~$s, \, t$ such that $s \approx_E^{} t$ and
$t$ is a proper subterm of~$s$. Equivalently, 
$\approx_E^{}$ is subterm-collapsing if and only if
$s \approx_E^{} s|_p^{}$ for some term~$s$ and~$p \in \Pos(s)$.
If the theory has a convergent term rewriting system~$R$, then
it is subterm-collapsing if and only if
$s \rightarrow_R^{+} s|_p^{}$ for some term~$s$ and~$p \in \Pos(s)$.
An equational theory is said to be \emph{non-subterm-collapsing}
or \emph{simple}~\cite{BHSS} if and only if
it is not subterm-collapsing.

Equational rewriting facilitates incorporation of equational
theories such as associativity and commutativity, which basic
(``pure'') term rewriting systems cannot handle, since they
cannot (often) be turned into \emph{terminating} rewrite rules. An
equational term rewriting system consists of a set of
identities~$E$ (which often contains identities such as
Commutativity and Associativity) and a set of rewrite
rules~$R$. This gives rise to a new rewrite relation
$\longrightarrow_{R,E}^{}$, which uses equational matching
modulo~$E$ instead of standard matching.

\vspace{0.05in}
\noindent
\underline{Example}: Let $E = \{(x+y)+z \approx x+(y+z), \, x + y \approx y + x \}$ and
$R = \{0+x \rightarrow x\}$. Then \[ (a+0)+b \; ~ \longrightarrow_{R,E}^{} ~ \; a+b \] since
$a + 0$ matches with~$0 + x$ modulo~$E$.\\

%\underline{\bf dag-solved form}:
\noindent
A set of equations is said to be in {\em dag-solved form\/} (or {\em
d-solved form\/}) if and only if they can be arranged as a
list  \[ X_1 =_{}^? t_1 , \; \ldots , \; X_n =_{}^? t_n \] where
  (a) each left-hand side $X_i$ is a distinct variable, and
  (b) $\forall \, 1 \le i \le j \le n$: $X_i$ does not occur in~$t_j$.

\noindent
\underline{\bf Equational Unification}:
Two terms $s$ and $t$ are unifiable modulo an
equational theory~$E$ iff there exists a substitution $\theta$ such that
$\theta (s) \; {{\approx}_E^{}} \; \theta (t)$. The unification
problem modulo equational theory~$E$ is the problem of solving a set of
equations $\cal S$ = $\{ s_1 \, {\approx}_{E^{}}^? \, t_1 , \ldots , s_n \,
{\approx}_{E^{}}^? \, t_n \} $, whether there exists $\sigma$ such that
$\sigma ( s_1 ) \allowbreak \, {{\approx}_E^{}} \allowbreak \, \sigma
( t_1 ) $, $\cdots \,$, $\sigma ( s_n ) \, {{\approx}_E^{}} \, \sigma
( t_n )$. This is also referred to as \emph{semantic} unification where
equational equivalence~\cite{BaaderSnyd-01} or congruence is considered among the terms
being unified, rather than syntactic identity.
Some of the standard equational theories used are \emph{associativity}
and \emph{commutativity}. 

A unifier $\delta$ is \emph{more general than} another unifer $\rho$ iff a
substitution equivalent to the latter can be obtained from the former
by suitably composing it with a third substitution:
\begin{center}
$ \delta\, {{\preceq}_E^{}} \, \rho \; \;$ iff $\; \; \exists \sigma :  \, \delta \circ \sigma \; {=_E^{}} \; \rho$
\end{center}

A substitution~$\theta$ is a \emph{normalized} substitution
with respect to a term rewrite system~$R$
if and only for every $x$, $\theta(x)$ is in
$R$-normal form. In other words, terms in
the range of~$\theta$ are in normal form. 
(These are also sometimes referred to as
as \emph{irreducible} substitutions.)
When~$R$ is convergent, one can assume that all unifiers
modulo~$R$ are normalized substitutions.

%\fxnote[marginclue]{This definition needs to be fixed.}

\pagebreak
\section{Asymmetric Unification}

\begin{defn}
Given a \emph{decomposition} $(\Sigma,E,R)$ of an equational theory, a
substitution~$\sigma$ is an asymmetric $R,E$-unifier of a set~$Q$ of
asymmetric equations $\{ s_1 \, \approx_{\smalldownarrow}^? \, t_1 , \, \ldots , \, s_n \,
\approx_{\smalldownarrow}^? \, t_n \} $ iff for each asymmetric equation $s_i \,
\approx_{\smalldownarrow}^? \, t_i$, $\sigma$ is an $(E \cup R)$-unifier of the equation
$s_i \, \approx^? \, t_i$, and $\sigma (t_i)$ is in $R,E$-normal form. In other words,
$\sigma (s_i)  \, \rightarrow_{R,E}^! \, \sigma (t_i)$.
\end{defn}

\noindent
(Note that symmetric unification can be reduced to asymmetric unification. We could also include 
symmetric equations in a problem instance.)\\
%A set of substitutions $\Omega$ is a \emph{complete} set of asymmetric $R,E$-unifiers of~$S$ iff:
%(i) every member of $\Omega$ is an asymmetric
%$R,E$-unifier of~$S$, and (ii) for every asymmetric $R,E$-unifier $\theta$
%of~$S$ there exists a $\sigma \in \Omega$ such that
%$\sigma{{\preceq}_E^{}} \theta$.\\

\noindent
\underline{\bf Example}:\\ Let $R=\{x+a \rightarrow x \}$ 
be a rewrite system. An asymmetric unifier $\theta$ for 
$\{u+v \, =_{\smalldownarrow}^? \, v+w\}$ 
modulo this system is
$\theta=\{u\mapsto v , \, w\mapsto v\}$. However, another unifier
$\rho=\{u\mapsto a , \, v\mapsto a, \, w\mapsto a\}$ is not
an asymmetric unifier. But note that $\theta\; {{\preceq}_E^{}} \; \rho $,
i.e., $\rho$ is an \emph{instance} of~$\theta$, or, alternatively,
$\theta$ is more general than~$\rho$.
This shows that {instances} of 
asymmetric unifiers need not be asymmetric unifiers.

\section{Disunification}

Disunification deals with solving a set of equations and disequations 
with respect to a given equational theory. 
\begin{defn}
For an equational theory $E$, a disunification problem is a set of equations and disequations $\cal L$ = 
$\{ s_1 \approx^?_{E^{}} t_1 , \ldots , s_n \approx^?_{E^{}} t_n \} \; ~ \bigcup ~ \;
\{s_{n+1}^{} \not\approx^?_{E^{}}
t_{n+1}^{} , \ldots , s_{n+m}^{} \not\approx^?_{E^{}} t_{n+m}^{}  \}$. 
\end{defn}

\noindent
A solution to this problem is a substitution $\sigma$ such
that: \[ \sigma (s_i) \; \approx_E^{} \; \sigma (t_i) \qquad (i=1, \ldots ,n) \] and \[ \sigma (s_{n+j}^{})  \; \not\approx_E^{} \; \sigma (t_{n+j}^{}) 
\qquad (j=1, \ldots , m). \] \\

\noindent
\underline{\bf Example}:\\ Given $E=\{x+a \approx
x \}$, a disunifier $\theta$ for $\{ u + v \, \not\approx_E \, v + u \}$
 is
$\theta=\{u\mapsto a, \, v \mapsto b \}$. 

If $a+x \approx x$ is added to the identities $E$, then $\theta =
\{ u\mapsto a, \, v \mapsto b \}$ is clearly no longer a
disunifier modulo this equational theory.

\pagebreak

\section{A theory for which asymmetric unification is in P whereas
  disunification is NP-complete}

\noindent
Let $R_1^{}$ be the following term rewriting system:
\begin{eqnarray*}
  h(a) & \rightarrow & f(a,c)\\
  h(b) & \rightarrow & f(b,c)
\end{eqnarray*}
\noindent
We show that asymmetric unifiability modulo this theory can be solved
in polynomial time. The algorithm is outlined in Appendix~$A$
(p.\ \pageref{appendix1}--\pageref{appendixAend}).

However, disunification modulo~$R_1^{}$ is NP-hard. The proof is by a
polynomial-time reduction from the three-satisfiability~(3SAT)
problem. 

Let $U=\{ x_1,x_2, \ldots, x_n\}$ be the set of variables, and
$B=\{C_1,C_2, \ldots, C_m\}$ be the set of clauses. Each clause~$C_k$, where
$1 \leq k \leq m$, has 3 literals.

%Question: Is there a truth assignment to $U$ satisfying all the clauses in $B$?

We construct an instance of a disunification problem from 3SAT.  There
are 8 different combinations of T and F assignments to the variables
in a clause in 3SAT, out of which there is exactly one
truth-assignment to the variables in the clause that makes the clause
evaluate to false.  For the 7 other
combinations of T and F assignments to the literals, the clause is
rendered true. We represent T by~\emph{a} and F by~\emph{b}.
Hence for each clause $C_i$ we create a disequation~$DEQ_i$
of the form \[ f(x_p,f(x_q,x_r)) \; \not\approx_{R_1^{}} \; f(d_1,f(d_2,d_3))
\] where $x_p,x_q,x_r$ are variables,
$d_1,d_2,d_3 \in \{a,b\} $, and $(d_1,d_2,d_3)$ corresponds to
the falsifying truth assignment. For example, given a clause
$C_k=x_p\vee \overline{x_q}\vee x_r$, we create the corresponding
disequation $DEQ_k \, = \, f(x_p,f(x_q,x_r)) \not\approx_{R_1^{}}
f(b,f(a,b))$.

We also create the equation $h(x_j)\approx_{R_1^{}}  f(x_j,c)$ for
each variable~$x_j$. These make sure that each $x_j$ is mapped to
either $a$ or~$b$.

Thus for $B$, the instance of disunification constructed is \[ S= \left\{
\vphantom{b^b} h(x_1)\approx f(x_1,c), \, h(x_2)\approx f(x_2,c), \, 
\ldots , \, h(x_n)\approx f(x_n,c) \right\} \; \cup \; \left\{
\vphantom{b^b} DEQ_1,DEQ_2, \ldots , DEQ_m \right\} \]

\noindent
\underline{\bf Example}: Given $U=\{x_1,x_2,x_3\}$ and $B=\{x_1 \vee
\overline{x_2} \vee x_3, \; \; \overline{x_1} \vee \overline{x_2}
\vee x_3 \}$, the constructed instance of
disunification is 
\begin{gather*}
\bigl\{ \vphantom{b^b} h(x_1)\approx f(x_1,c), \; h(x_2)\approx f(x_2,c),
\; h(x_3)\approx f(x_3,c), \; 
f(x_1,f(x_2,x_3)) \not\approx f(b,f(a,b)) , \\
f(x_1,f(x_2,x_3))
\not\approx f(a,f(a,b)) \bigr\}
\end{gather*}

Note that membership in NP is not hard to show since $R_1^{}$ is saturated by paramodulation~\cite{DBLP:conf/cade/LynchM02}.
%~\cite{Bouchard2013} and hence has the Finite Variant Property (FVP). 

\pagebreak

\section{A theory for which disunification is in P whereas
  asymmetric unification is NP-hard}

\noindent
The theory we consider consists of the following term rewriting system~$R_2^{}$:
\begin{eqnarray*}
  x + x & \rightarrow & 0\\
  x + 0 & \rightarrow & x\\
  x + (y + x) & \rightarrow & y
\end{eqnarray*}
and the equational theory~$AC$:
\begin{eqnarray*}
  (x + y) + z & \approx & x + (y + z)\\
  x + y & \approx & y + x
\end{eqnarray*}
This theory is called \textbf{ACUN} because it consists of
\emph{associativity, commutativity, unit} and \emph{nilpotence}. This
is the theory of the boolean XOR operator. An algorithm for 
\emph{general} \textbf{ACUN} unification is provided by Zhiqiang
Liu~\cite{Zhiqiang-Liu} in his Ph.D.~dissertation. (See also~\cite[Section~4]{DBLP:conf/cade/ErbaturEKLLMMNSS13}.)

Disunification modulo this theory can be solved in polynomial time by
what is essentially Gaussian Elimination over $\mathbb{Z}_2$.

Suppose we have $m$ variables $x_1,x_2, \ldots , x_m,$ and $n$
constant symbols $c_1,c_2,\ldots,c_n,$ and $q$~such equations
and disequations to be unified. We can assume an ordering on the
variables and constants
$x_1 > x_2 > \ldots > x_m > c_1 > c_2 > \ldots > c_n$. 
We first pick an equation with leading
variable~$x_1$ and eliminate $x_1$ from all \emph{other} equations and disequations. We
continue this process with the next equation consisting of 
leading variable~$x_2$, followed by an equation containing
leading variable~$x_3$ and so on, until no more variables can
be eliminated. The problem has a solution if and only if
$(i)$
there are no
equations that contain only constants, such as $c_3+c_4
\approx c_5$, and
$(ii)$ there are no disequations of the form~$0 \not \approx
0$.  This way we can solve the disunification problem in
polynomial time using Gaussian Elimination over $\mathbb{Z}_2$
technique. \\

\noindent
\underline{\bf Example}: Suppose we have two equations
$x_{1} + x_{2} + x_{3} + c_1 + c_2 \; \approx^?_{{R_2^{}, AC}^{}} \; 0 $ and
$x_{1} + x_{3} + c_2 + c_3 \;  \approx^?_{{R_2^{}, AC}^{}} \; 0 $, and a disequation 
$x_{2} \; {\not\approx}^?_{{R_2^{}, AC}^{}} \; 0$.
%\begin{equation}
%A = 

%\left({\begin{array}{cccccc} 1 & 1 & 1 & 1 & 1 & 0 \\ 1 & 0 & 1 & 0 & 1 & 1\\%\end{array}}\right)
%\left({\begin{array}{c} x_1 \\ x_2 \\ x_3 \\ c_1 \\ c_2 \\ c_3 \\ \end{array}}\right)
%=\left({\begin{array}{c}  0 \\ 0\\ \end{array}}\right)
%\end{equation}\\

Eliminating $x_1$ from the second equation, results in
the equation~$x_2 + c_1 + c_3 \; \approx_{R_2^{}, AC}^{} \; 0$. 
We can now eliminate~$x_2$ from the first equation, 
resulting in~$x_{1} + x_{3} + c_2 + c_3 \; \approx_{R_2^{}, AC}^{} \; 0$.
$x_2$~can also be eliminated from the disequation~$x_{2} \; {\not\approx}_{R_2^{}, AC}^{} \; 0$, which gives us~$c_1 + c_3 \; {\not\approx}_{R_2^{}, AC}^{} \; 0$.
Thus
the procedure terminates with 
\begin{eqnarray*}
x_{1} + x_{3} + c_2 + c_3 & \approx_{R_2^{}, AC}^{} & 0\\
x_{2} + c_1 + c_3 & \approx_{R_2^{}, AC}^{} & 0\\
c_1 + c_3 & {\not\approx}_{R_2^{}, AC}^{} & 0
\end{eqnarray*}
Thus we get
\begin{eqnarray*}
x_{2} & \approx_{R_2^{}, AC}^{} & c_1 + c_3\\
x_{1} + x_{3} & \approx_{R_2^{}, AC}^{} & c_2 + c_3\\
\end{eqnarray*}
and the following substitution is clearly a solution: \[ \left\{
\vphantom{b^b} 
x_1 \mapsto c_2, \;
x_{2} \mapsto c_1 + c_3, \;
x_3 \mapsto c_3 \right\} \]

%\pagebreak

However, asymmetric unification is NP-hard. The proof is by a
polynomial-time reduction from the graph 3-colorability problem.

Let $G=(V,E)$ be a graph where $V=\{v_1,v_2,v_3, \ldots, v_n\}$ are the
vertices, $E=\{e_1,e_2,e_3, \ldots, e_m\}$ the edges and
$C=\{c_1,c_2,c_3\}$ the color set with $n\geq 3$.
$G$~is 3-colorable if none of the adjacent vertices $\{v_i,v_j\}\in E$
have the same color assigned from~$C$. 
We construct an
instance of asymmetric unification as follows.
We create variables for vertices and edges in~$G$: for each vertex~$v_i$
we assign a variable~$y_i$ and for each edge~$e_k$ we
assign a variable~$z_k$. Now for every edge
$e_k=\{v_i,v_j\}$ we create an equation $EQ_k = c_1 + c_2 + c_3
\approx^?_{\downarrow} y_i + y_j + z_k$. Note that
each $z_k$ appears in only one equation.

Thus for $E$, the instance of asymmetric unification problem
constructed is \[ S = \left\{ \vphantom{b^b} EQ_1, \, EQ_2, \, \ldots , \, EQ_m \right\} \]

If $G$ is 3-colorable, then there is a color assignment $\theta: V
\rightarrow C$ such that $\theta v_i \neq \theta v_j$ if
$e_k=\{v_i,v_j\}\in E$. This can be converted into an asymmetric
unifier~$\alpha$ for~$S$ as follows:
We assign the color of~$v_i$, $\theta (v_i)$ to~$y_i$, $\theta(v_j)$ to~$y_j$,
and the remaining color to~$z_k$. Thus
$\alpha (v_i+v_j+z_k) \approx_{AC}^{} c_1 + c_2 + c_3$ and therefore
$\alpha$ is an asymmetric unifier of~$S$. 
Note that the term
$c_1 + c_2 + c_3$ is clearly in normal form
modulo the rewrite relation~$\longrightarrow_{R_2^{}, AC}^{}$.

Suppose $S$ has an asymmetric unifier~$\beta$. 
Note that $\beta$ cannot map $y_i, \; y_j$ or~$z_k$
to $0$ or to a term of the form $u + v$ since
$\beta( y_i+y_j+z_k )$ has to be in normal form or
irreducible. Hence for each equation $EQ_k$, it must be
that $\beta (y_i), \beta (y_j), \beta (z_k)
\in \{c_1,c_2,c_3\}$ and $\beta (y_i) \not = \beta (y_j) \not = \beta (z_k)$.
Thus
$\beta$ is a 3-coloring of~$G$.  \\

\noindent
\underline{\bf Example}: Given $G=(V,E), V=\{v_1,v_2,v_3,v_4\}$,
$E=\{e_1,e_2,e_3,e_4\}$, where
$e_1 =\{v_1,v_3\}, \; e_2=\{v_1,v_2\}, \; e_3=\{v_2,v_3\}, \; e_4=\{v_3,v_4\}$ and
$C=\{c_1,c_2,c_3\}$, the constructed instance of asymmetric
unification is
\begin{center}
$EQ_1 ~ = ~ c_1 + c_2 + c_3 \; \;  \approx^?_{\downarrow} \; \;  y_1 + y_3 +
z_1$\\ $EQ_2 ~ = ~ c_1 + c_2 + c_3 \; \;  \approx^?_{\downarrow} \; \;  y_1 + y_2 +
z_2$\\ $EQ_3 ~ = ~ c_1 + c_2 + c_3 \; \;  \approx^?_{\downarrow} \; \;  y_2 + y_3 +
z_3$\\ $EQ_4 ~ = ~ c_1 + c_2 + c_3 \; \;  \approx^?_{\downarrow} \; \;  y_3 + y_4 + z_4$.\\
\end{center}

Now suppose the vertices in the graph $G$ are given this color
assignment: $\theta = \{v_1\mapsto c_1, v_2 \mapsto c_2, v_3\mapsto c_3,
v_4\mapsto c_1\}$. We can create an asymmetric
unifier based on this~$\theta$ by mapping each $v_i$ to~$\theta(v_i)$
and, for each edge~$e_j$, mapping $z_j$ to
the remaining color from~$\{c_1,c_2,c_3\}$
after both its vertices are assigned. For instance,
for $e_1 =\{v_1,v_3\}$, since $y_1$ is mapped to $c_1$ and
$y_3$ is mapped to~$c_2$, we have to map
$z_1$ to~$c_3$. Similarly for $e_2=\{v_1,v_2\}$, we map
$z_2$ to~$c_2$ since $y_1$ is mapped to $c_1$ and
$y_2$ is mapped to~$c_3$. Thus the asymmetric unifier is \[
\left\{ \vphantom{b^b}
y_1\mapsto c_1, \; y_2\mapsto c_3, \;
y_3\mapsto c_2, \; z_1\mapsto c_3, \; z_2\mapsto c_2, \;
z_3\mapsto c_1, \;
z_4\mapsto c_3
\right\} \]

\noindent
We have not yet looked into whether the problem is in~NP, but we expect it to be so.

\pagebreak
\section{A theory for which ground disunifiability is in P whereas
  asymmetric unification is NP-hard}

%\todo{Title needs to be fixed.}
This theory is the same as the one mentioned in previous section,
\textbf{ACUN}, but with a homomorphism added. It has an $AC$-convergent
term rewriting system, which we call~$R_3^{}$:

\vspace{0.05in}
\noindent
%Let $R_3^{}$ be the following term rewriting system:
\begin{eqnarray*}
  x + x & \rightarrow & 0\\
  x + 0 & \rightarrow & x\\
  x + (y + x) & \rightarrow & y\\
%\end{eqnarray*}
%with homomorphism
%\begin{eqnarray*}
  h(x+y) & \rightarrow & h(x) + h(y)\\
  h(0) & \rightarrow & 0
\end{eqnarray*}
%and the equational theory~$AC$:
%\begin{eqnarray*}
  %(x + y) + z & \approx & x + (y + z)\\
  %x + y & \approx & y + x
%\end{eqnarray*}

\noindent
\subsection{Ground disunification} 
Ground disunifiability~\cite{BaaderSchulzDisunif95} problem refers to checking for
ground solutions for 
a set of disequations
and equations. The restriction is that only the set of constants provided in the input,
i.e., the equational theory and the equations and disequations,
can be used; no new constants can be introduced.

We show that ground disunifiability modulo this theory can be solved
in polynomial time, by reducing the problem to that of solving systems
of linear equations. This involves finding the Smith Normal
Form~\cite{Greenwell09solvinglinear,KANNAN198569,kaltofen1987fast}. This
gives us a general solution to all the variables or unknowns.

Suppose we have $m$ equations in our ground disunifiability
problem. We can assume without loss of generality that the
disequations are of the form~$z \neq 0$. For example, if we have
disequations of the form $e_1 \neq e_2$, we introduce a new variable
$z$ and set $z = e_1 + e_2$ and $z \neq 0$. Let $n$ be the number of
variables or unknowns for which we have to find a solution.

%For ground disunifiability we need to solve over all constants. 

For each constant in our ground disunifiability problem, we follow the
approach similar to~\cite{guo2000complexity}, of forming a set of linear
equations and solving them to find ground solutions. \\

We use $h^k x$ to represent the
term $h(h( \ldots h(x) \ldots ))$ and $H^k = h^{k_1}x + h^{k_2}x + \cdots +
h^{k_n}x$ is a polynomial over $\mathbb{Z}_2[h]$.\\[-5pt]

%Step 1: 
We have  \begin{center}$s_i = H_{i1}x_{1} + H_{i2}x_{2} + \ldots +H_{im}x_{n}, \; \; H_{ij} \in \mathbb{Z}_2[h] $\end{center}
 \begin{center}$t_i = H^{'}_{i1} c_1^{} + H^{'}_{i2} c_2^{} + \ldots +H^{'}_{im} c_l^{} , \; \; H^{'}_{ij} \in \mathbb{Z}_2[h] $\end{center}
where \begin{enumerate}
\item[] $\qquad \{c_1^{}, \ldots c_l^{}\}$ is the set of constants and\\[-18pt]
\item[] $\qquad \{x_1, \ldots x_n\}$ is the set of variables.
\end{enumerate}

For each constant $c_i, 1\leq i \leq l$, and each variable~$x$, we
create a variable $x^{c_i}$.
We then 
generate, for each constant~$c_i$, a set of linear equations $S^{c_i}$ 
of the form $AX =^?  B$ with coefficients from the polynomial
ring~$\mathbb{Z}_2[h]$.

% Solve $S^{c_i}$ over $\mathbb{Z}_2[h]$.\\ If we have a disequation of
% the form $z \neq 0$, we would have to check for all $c_i's$ to be 0,
% otherwise we do not have a solution.\\
 
The solutions are found 
by computing the Smith Normal
Form of~$A$. We now outline that procedure\footnote{We follow the notation and procedure similar to Greenwell and Kertzner~\cite{Greenwell09solvinglinear}}:

Note that the dimension of matrix $A$ is $m \times n$
where $m$ is the number of equations and 
$n$ is the number of unknowns. The dimension of
of matrix $B$ is $ m \times 1$. Every matrix $A$, of rank~$r$, is
equivalent to a diagonal matrix~$D$, given by
\begin{center}$D = diag(d_{11}, d_{22}, \ldots d_{rr}, 0, \ldots, 0)$\end{center}
Each entry $d_{kk}^{}$ is different from 0 and 
the entries form a divisibility sequence. \\

The diagonal matrix $D$, of size $m \times n$, is the Smith Normal
Form (SNF) of matrix~$A$.  There exist invertible \emph{matrices} $P$, of size
$m \times m$, and $Q$, of size $n \times n$ such that

\begin{equation}D = PAQ \end{equation}

and let \begin{center}$\overline{D} = diag(d_{11}, d_{22}, \ldots, d_{rr})$\end{center}
be the submatrix
 consisting of the first $r$~rows and
the first $r$~columns of~$D$.\\

\noindent
%\textbf{Theorem}: 
Suppose $AX = B$. We have, from $(1)$,
\begin{center}$PAX = PB$\end{center} 
Since $Q$ is invertible we can write
\begin{center}$PAQ (Q^{-1}X) = PB$ \end{center}
Let $C = PB$ and 
\begin{center}$ Y = (Q^{-1}X) = 
\begin{bmatrix}\; \overline{Y} \; \\ Z \end{bmatrix}$ \end{center}
with $\overline{Y}$ being first $r$ rows of the $n \times 1$ matrix $Y$, 
and $Z$ the remaining $(n-r)$ rows of $Y$.

\begin{center} $C \text{ can be written as }
\begin{bmatrix} \; \overline{C} \; \\ U \end{bmatrix}$ \end{center}
with $\overline{C}$ the first $r$ rows of $C$, and $U$ a matrix of zeros.\\

Then $DY = PB = C$ translates into\\
\begin{equation*}
\begin{bmatrix} \; \overline{D} & 0 \; \\ 0 & 0 \end{bmatrix}
\begin{bmatrix} \; \overline{Y} \; \\ Z \end{bmatrix}
= \begin{bmatrix} \; \overline{C} \; \\ U \end{bmatrix}
\end{equation*}

%Thus we also have that \begin{center} $X = QY$ \end{center}
We solve for $Y$ in $DY = C$, by first solving 
$\overline{D} \, \overline{Y} = \overline{C}$:
\begin{equation*}
%R^2 = 
\begin{bmatrix}d_{11} & & \\ & \ddots & \\ & & d_{rr}\end{bmatrix}
\begin{bmatrix} y_1\\ y_2 \\ y_3\\ \vdots \\y_r \end{bmatrix}
= \begin{bmatrix} c_1\\ c_2 \\ c_3\\ \vdots \\c_r \end{bmatrix}
%= c^2 + s^2
\end{equation*}
A solution exists if and only if each $d_{ii}$ divides $c_i$. If this
is the case let $\widehat{y_i} = c_i / d_{ii}$. Now to find a general
solution plug in values of~$Y$ in~$X = QY:$\\

%\begin{equation*}
%Q 
%\begin{bmatrix} \widehat{y_1}\\ \widehat{y_2} \\ \ldots \\ \widehat{y_r} \\ z_{r+1}\\ \ldots \\ z_n \end{bmatrix}
%\end{equation*}
\[
\begin{blockarray}{cccccc}
      & \Scale[1]{r}  & & & \Scale[1]{n - r} &\\
    \begin{block}{[ccc|ccc]}
      &                   &  & &   &   \\
      &                   &   &     & &          \\
      & \Scale[1]{Q_1^{}}  &   &   & \Scale[1]{Q_2^{}}    &      \\
      &                   &   &   & & \\
      &                   &   & & & \\
    \end{block}
\end{blockarray}
\begin{bmatrix} \widehat{y_1}\\ \widehat{y_2} \\ \vdots \\ \widehat{y_r} \\ z_{r+1}\\ \vdots \\ z_n \end{bmatrix}
\]

First $r$ columns of $Q$ are referred to as $Q_1$ and remaining $n-r$
columns are referred to as $Q_2$. To find a particular solution,
for any~$x_j$,  we take the dot product of the $j_{}^{\mathrm{th}}$~row of~$Q_1$ and 
$\left( \widehat{y_1}, \ldots , \widehat{y_r} \right)$. 

Similarly,
to find a general solution, we take the dot product of $i^{th}$ row of
$Q_1$ with $\left( \widehat{y_1}, \ldots ,\widehat{y_r} \right)$, plus the dot product of 
the $i^{th}$ row of~$Q_2$, with a vector 
$\left( z_{r+1}, \ldots , z_n \right)$ consisting of distinct variables.

If we have a disequation of the form $x_i \neq 0$, to check for 
solvability for~$x_i$, we first check whether the ~particular~ solution is~$0$. 
If it is not, then we are done. Otherwise, check whether
all the values in $i^{th}$ row of
$Q_2$ are identically~0. If it is not, then we
have a solution since $z_{r+1}, \ldots , z_n$ can take any
arbitrary values.
This procedure has to be repeated for all constants.\\
 
\subsection{Ground Asymmetric Unification} 
However, asymmetric unification modulo~$R_3^{}$ is NP-hard.
Decidability can be shown by automata-theoretic methods as for Weak
Second Order Theory of One successor
(WS1S)~\cite{Elgot,buchi1960weak}.\\

In WS1S we consider quantification over finite sets of natural
numbers, along with one successor function. All equations or formulas
are transformed into finite-state automata which accepts the strings that
correspond to a model of the
formula~\cite{klaedtke2002parikh,vardi2008automata}. This automata-based
approach is key to showing decidability of~WS1S, since the
satisfiability of WS1S formulas reduces to the automata
intersection-emptiness problem. We follow the same approach here. 

For ease of exposition, let us
consider the case where there is only one constant~$a$. Thus every
ground term can be represented as a set of natural numbers. The
homomorphism~$\mathsf{h}$ is treated as a successor
function. Just as in WS1S, the input to the automata are column vectors of bits. The length of each column vector
is the number of variables in the problem.

\begin{center} $\Sigma=\left\{ \vphantom{b^b}
          \begin{psmallmatrix}
           0 \\
           0 \\
          \vdots \\
           0 
          \end{psmallmatrix}, \ldots ,\begin{psmallmatrix}
           1 \\
           1 \\
           \vdots \\
           1 
          \end{psmallmatrix} \right\}$\end{center} 

The deterministic finite automata (DFA) are illustrated in
Appendix~$C$ (p.\ \pageref{appendix3}--\pageref{appendix4}). 
The $\mathsf{+}$~operator behaves like the 
\emph{symmetric set difference} operator.\\

Once we have automata constructed for all the formulas, we take the
intersection and check if there exists a string accepted by
corresponding automata. If the intersection is not empty, then we have a
solution or an asymmetric unifier for set of formulas.

This technique can be extended to the case
where we have more than one constant. 
Suppose we have $k$ constants, say~$c_1 , \ldots , c_k$.
We express each variable~$X$ in terms of the constants as follows: \[
X ~ = ~ X_{}^{c_1} + \ldots + X_{}^{c_k} \] effectively grouping
subterms that contain each constant under a new variable. Thus
if $X = h^2 (c_1) + c_1 + h(c_3)$, then 
$X_{}^{c_1} = h^2 (c_1) + c_1$, $X_{}^{c_2} = 0$, and
$X_{}^{c_3} = h(c_3)$. If the variables are~$X_1, \ldots , X_m$,
then we set
\begin{align*}X_1 &= X_1^{c_1}+ \ldots +X_1^{c_k} \\X_2 &= X_2^{c_1}+  \ldots  +X_2^{c_k} \\ &\vdots \\ X_m &= X_m^{c_1}+ \ldots  +X_m^{c_k}\end{align*}

For example, if $\mathsf{Y}$ and $\mathsf{Z}$ are set variables and
$a,b,c$ are constants, then we can write $\mathsf{Y}$ = $\mathsf{Y^a+Y^b+Y^c}$
and $\mathsf{Z = Z^a+Z^b+Z^c}$ as our terms with constants.
For each original variable, say~$Z$, we refer to 
$Z^{c_1}$ etc.\ as its \emph{components} for ease of exposition.

If the equation to be solved is $:\mathsf{X= h(Y)}$, with
$\mathsf{a,b,c}$ as constants, then we create the equations
$\mathsf{X^a = h(Y^a)}$, $\mathsf{X^b = h(Y^b)}$, 
$\mathsf{X^c = h(Y^c)}$. However, if the equation
is asymmetric, i.e., $\mathsf{X =_{\downarrow} h(Y)}$, then
$\mathsf{Y}$ has to be a term of the form~$h_{}^i (d)$
where $d$~is either $a$, $b$, or~$c$.
All components except one have to be $0$
% Thus we have to nondeterministically pick this~$d$
%and
 and we form the equation~$\mathsf{X^d =_{\downarrow}^{} h(Y^d)}$ since $\mathsf{Y \not= 0}$.
The other components for $\mathsf{X}$ and $\mathsf{Y}$
have to be~0.

Similarly, if the equation to be solved is $\mathsf{X = W+Z}$, with $\mathsf{a,b,c}$ as constants,
we form the equations $\mathsf{X^a=W^a+Z^a}$, $\mathsf{X^b=W^b+Z^b}$ and
$\mathsf{X^c=W^c+Z^c}$ and solve the equations. If we have an asymmetric equation 
$\mathsf{X =_{\downarrow} W+Z}$, 
then clearly one of
the components of each original variable has to be non-zero; e.g., in
$\mathsf{W=W^a+W^b+W^c}$, all the components cannot be~$0$
simultaneously. 
It is ok for $\mathsf{W^a}$ and
$\mathsf{Z^a}$ to be~$0$ simultaneously, provided either 
one of
$\mathsf{W^b} \text{ or } \mathsf{W^c}$  is non-zero \emph{and}
one of 
$\mathsf{Z^b} \text{ or } \mathsf{Z^c}$, is non-zero. 
For example, $\mathsf{W=W^b}$ and $\mathsf{Z=Z^c}$ is fine, 
i.e, $\mathsf{W}$ can be equal to its
$\mathsf{b}$-component and $\mathsf{Z}$ can be 
equal to its $\mathsf{c}$-component,
respectively, as in the solution $\left\{ \vphantom{b^d}
W \mapsto h^2(b) + h(b), ~ Z \mapsto h(c) + c, ~
X \mapsto h^2(b) + h(b) + h(c) + c \right\}$.
If $\mathsf{W^a}$ and $\mathsf{Z^a}$ are
non-zero, they cannot have anything in common, or otherwise there will be
a reduction. In other words, $\mathsf{X^a}$, $\mathsf{W^a}$ and $\mathsf{Z^a}$ 
must be solutions of the asymmetric equation~$\mathsf{X^a} =_{\downarrow}^{} 
\mathsf{W^a} + \mathsf{Z^a}$.\\

Our approach is to design a nondeterministic
algorithm. We 
guess which constant component in each variable has to be~$0$, i.e.,
for each variable~$X$ and each constant~$a$, we ``flip a coin''
as to whether $X^a$ will be set equal to~$0$ by the target solution. Now for the case
$\mathsf{X =_{\downarrow} W+Z}$, we do the following:
\begin{quote}
\begin{tabbing}
for \= all constants $a$ do:\\
    \> if $\mathsf{X^a = W^a = Z^a = 0}$ then skip\\
    \> else \= if $\mathsf{W^a = 0}$ then set $\mathsf{X^a = Z^a}$\\
    \>      \> if $\mathsf{Z^a = 0}$ then set $\mathsf{X^a = W^a}$\\
    \>      \> if both $\mathsf{W^a}$ and $\mathsf{Z^a}$ are non-zero then set
$\; \mathsf{X^a =_\downarrow W^a + Z^a}$
\end{tabbing}
\end{quote}

In the asymmetric case $\mathsf{X =_{\downarrow} h(Y)}$, if more than one
of the components of~$\mathsf{Y}$ happens to be~non-zero, it is clearly an error.
(``The guess didn't work.''). Otherwise, i.e., if exactly one of the 
components is non-zero, we form
the asymmetric equation as described above. \\

\noindent
$\underline{\emph{Nondeterministic Algorithm when we have more than one constant}}$
\begin{enumerate}
\item If there are $m$ variables and $k$ constants, then 
  represent each variable in terms of its $k$~constant components.
\item Guess which constant components have to be~$0$.
\item Form symmetric and asymmetric equations for each constant.
\item Solve each set of equations by the Deterministic Finite Automata
  (DFA) construction.
\end{enumerate}

\noindent
The exact complexity of this problem is open.

%\noindent
    
\pagebreak
\section{A theory for which time complexity of Asymmetric Unification varies based on ordering of function symbols}
\noindent
Let $E_4^{}$ be the following equational theory:
\begin{eqnarray*}
  g(a) & \approx & f(a,a,a)\\
  g(b) & \approx & f(b,b,b)
\end{eqnarray*}
\noindent

\noindent
Let $R_4$ denote 
\noindent
\begin{eqnarray*}
  f(a,a,a) & \rightarrow & g(a)\\
  f(b,b,b) & \rightarrow & g(b)
\end{eqnarray*}
\noindent
This is clearly terminating, as can be easily shown by the 
\emph{lexicographic path ordering (lpo)}~\cite{Term} using
the symbol ordering~$f > g > a > b$.
We show that asymmetric unification modulo the rewriting system $R_4^{}$ is NP-complete. 
The proof is by a
polynomial-time reduction from the Not-All-Equal Three-Satisfiability ~(NAE-3SAT)
problem~\cite{DBLP:conf/rta/BrahmakshatriyaDGN13}. 

Let $U=\{ x_1,x_2, \ldots, x_n\}$ be the set of variables, and
$C=\{C_1,C_2, \ldots, C_m\}$ be the set of clauses. Each clause~$C_k$,
has to have at least one \emph{true} literal and at least one \emph{false}
literal.

We create an instance of asymmetric unification as follows. We
represent T by~\emph{a} and F by~\emph{b}. For each variable~$x_i$ we
create the equation \begin{center}$f(x_i, x_i, x_i)\approx_{R_4^{}}
  g(x_i)$\end{center} These make sure that each $x_i$ is mapped to
either $a$ or~$b$. For each clause $C_j=x_p\vee x_q \vee x_r$, we
introduce a new variable $z_j$ and create an asymmetric equation
$EQ_j:$
\begin{center} $z_j
\approx^?_{\downarrow} f(x_p, x_q, x_r)$ \end{center}
%The $R.H.S$ of equation $EQ_j$ is in normal form or irreducible modulo the rewrite system $R_4^{}$. We assign an arbitrary variable to the $L.H.S$.

Thus for any~$C$, the instance of asymmetric unification problem
constructed is  \[ \mathcal{S}= \left\{
\vphantom{b^b} f(x_1, x_1, x_1) \approx g(x_1), \,
\ldots , \, f(x_n, x_n, x_n) \approx g(x_n) \right\} \; \cup \; \left\{
\vphantom{b^b} EQ_1,EQ_2, \ldots , EQ_m \right\} \] 

If $\mathcal{S}$ has an asymmetric unifier $\gamma$, then, $x_p, x_q$ and $x_r$
cannot map to all $a$'s or all $b$'s since these will cause a reduction.
Hence for $EQ_j$,
$\gamma(x_p)$, $\gamma(x_q)$ and $\gamma(x_r)$ should take at least
one $\emph{a}$ and at least one $\emph{b}$. Thus $\gamma$ is also
a solution for NAE-3SAT.

Suppose, all clauses in $C$ have a satisfying assignment. Then $\{x_p,
x_q, x_r\}$ cannot all be~T or all~F, i.e., $\{x_p, x_q,
x_r\}$ needs to have at least one true literal and at least
one false literal. Thus if $\sigma$ is a satisfying assignment, we can
convert $\sigma$ into an asymmetric unifier~$\theta$ as follows:
$\theta(x_p) := \sigma(x_p)$, the value of 
$\sigma(x_p)$, $\emph{a}$ or $\emph{b}$, is assigned to $\theta(x_p)$. 
Similarly $\theta(x_q) := \sigma(x_q)$ and
$\theta(x_r) := \sigma(x_r)$. Recall that we also introduce a unique variable~$z_j$
for each clause~$C_j$ in~$C$.  Thus
if $C_j = \{ x_p, x_q, x_r \}$ we can map $z_j$
to~$\theta {(f(x_p, x_q, x_r))}$.  Thus
$\theta$ is an asymmetric unifier of~$S$ and
$z_j~\approx^?_{\downarrow}~f(x_p, x_q, x_r)$.~Note that $f(x_p, x_q, x_r)$
is clearly in normal form modulo the rewrite
relation~$\longrightarrow_{R_4^{}}^{}$, since $x_p,\;x_q,\;x_r$ can't all be same.

\vspace{0.1in}
\noindent
\underline{\bf Example}: Given $U=\{x_1, x_2, x_3, x_4\}$ and $C=\{x_1 \vee
x_2 \vee x_3, \; \; x_1 \vee x_2
\vee x_4, \; \; x_1 \vee x_3
\vee x_4 , \; \; x_2 \vee x_3
\vee x_4 \}$ the constructed instance of
asymmetric unification $\mathcal{S}$ is 
\begin{gather*}
\left\{ \vphantom{b^b} f(x_1, x_1, x_1)\approx g(x_1), \; f(x_2, x_2, x_2)\approx g(x_2),
\; f(x_3, x_3, x_3)\approx g(x_3), \; f(x_4, x_4, x_4)\approx g(x_4), \; \right. \\
z_1 \; \;  \approx^?_{\downarrow} \; \;  f(x_1 , x_2 , x_3), \\
z_2 \; \;  \approx^?_{\downarrow} \; \;  f(x_1 , x_2 , x_4), \\
z_3 \; \;  \approx^?_{\downarrow} \; \;  f(x_1 , x_3 , x_4), \\
\left. \vphantom{b^b}
z_4 \; \;  \approx^?_{\downarrow} \; \;  f(x_2 , x_3 , x_4)  \right\}
\end{gather*}
Again, membership in NP can be shown using the fact
that $R_4^{}$ is saturated by paramodulation~\cite{DBLP:conf/cade/LynchM02}% is forward-closed~\cite{Bouchard2013} and hence has the Finite Variant Property (FVP). \\

However, if we orient the rules the other way, i.e., when $g > f > a > b$, 
we can show that asymmetric unifiability modulo this theory can be solved
in polynomial time, i.e., when the 
term rewriting system is
\begin{eqnarray*}
  g(a) & \rightarrow & f(a,a,a)\\
  g(b) & \rightarrow & f(b,b,b)
\end{eqnarray*}
Let $R_5^{}$ denote the above term rewriting system.
The algorithm is outlined in Appendix~$B$
(p.\ \pageref{appendix2}--\pageref{appendixBend}). \\
\pagebreak
\bibliography{unif-asymm-dis-copy3}
\bibliographystyle{plain}

\pagebreak

\appendix

\renewcommand\thesection{\Alph{section}}
\section{Asymmetric Unifiability modulo $R_1^{}$}

\label{appendix1}

\noindent
Recall that the term rewriting system~$R_1^{}$ is
\begin{eqnarray*}
h(a) & \rightarrow & f(a, c) \\
h(b) & \rightarrow & f(b, c)
\end{eqnarray*}

Note that reversing the directions of the rules also produces a 
convergent system, i.e.,
\begin{eqnarray*}
f(a, c) & \rightarrow & h(a) \\
f(b, c) & \rightarrow & h(b)
\end{eqnarray*}
is also terminating and confluent.
We assume that the input equations are in standard form, i.e.,
of one of four kinds: $X \approx_{}^? Y$, $X \approx_{}^? h(Y)$,
$X \approx_{}^? f(Y, Z)$ and $X \approx_{}^? d$ where $X, Y, Z$~are variables
and $d$~is any constant. Asymmetric equations will have the extra downarrow,
e.g., $X \approx_{\downarrow}^{?} h(Z)$.\\

Our algorithm transforms an asymmetric unification problem to 
a set of equations in \emph{dag-solved form}
along
with \emph{clausal constraints,} where each atom 
is of the form \mbox{$( \langle variable \rangle \; = \; \langle constant \rangle )$}. 
We use the notation
$EQ \; \parallel \; \Gamma$, where $EQ$ is set of equations in
standard form as mentioned above, and $\Gamma$~is a set of clausal
constraints. Initially $\Gamma$ is empty.\\

\noindent
\begin{appxlem}  (Removing asymmetry) 
If $s$ is an irreducible term, then $h(s)$ is $($also$)$ irreducible
iff $s \not = a $ and $s \not = b$.
\end{appxlem}

\begin{proof}
If $s = a$ or $s = b$, then clearly $h(s)$ is reducible.
Conversely, if $s$ is irreducible and $h(s)$ is
reducible, then $s$ has to be either~$a$ (for the 
first rule to apply) or~$b$ (for the second rule).
\end{proof}

Hence we first apply the following inference rule (until finished)
that gets rid of asymmetry:\\

\begin{tabular}{lcc}
 &  & $\vcenter{
\infer{\eq ~ \uplus ~ \{X \approx_{}^? h(Y)\} \; \parallel \; \Gamma ~ \cup ~ \{
\neg(Y = a) \} ~ \cup ~ \{\neg(Y = b) \}  }
      { \eq ~ \uplus ~ \{X \approx_{\downarrow}^? h(Y)\} \; \parallel \; \Gamma}
}
$\\[+30pt]
\end{tabular}

%\noindent\underline{ \bf Lemma 2}: Variable Replacement: If $X$ occurs in $\eq$, replace $X$ by $V$ in all equations $\eq$ and clauses $\Gamma$.\\

\noindent 
\begin{appxlem} (Cancellativity) 
$h(s)$ $\downarrow_{R_1}$ $h(t)$ iff $s$ $\downarrow_{R_1}$
$t$. Similarly, $f(s_1, s_2)$ $\downarrow_{R_1}$ $f(t_1, t_2)$ iff
$s_1$ $\downarrow_{R_1}$~$t_1$ and $s_2$ $\downarrow_{R_1}$ $t_2$.
\end{appxlem}
%\noindent\underline{ \bf Lemma 4}:  Cancellativity: $f(X, Y)$ is equivalent to $f(W, T)$ iff $X$ is equivalent to $W$ modulo the rewrite system $R_1^{}$ and $Y$ is equivalent to $T$ modulo $R_1^{}$. 

\noindent
\begin{proof} The $\emph{if}$ part is straightforward. If $s$
and $t$ are joinable, this implies $h(s)$ and $h(t)$ are joinable
modulo $R_1$.\\ 

$\emph{Only\;if}$ part: Suppose $h(s)$ is joinable
with $h(t)$. Without loss of generality assume $s$ and $t$ are in
normal form. If $s=t$ then we are done. Otherwise, if $s \not = t$,
since we assumed $s$ and $t$ are in normal forms, $h(s)$ or $h(t)$
must be reducible. If $h(s)$ is
reducible, then $s$ has to be either $a$ or $b$, which reduces $h(s)$
to $f(s, c)$. Then $h(t)$ must also be reducible and joinable with
$f(s, c)$. Hence $s$ and $t$ will be equivalent.\\
%, since $t$ is in normal form, $t$ should be either $a$ or $b$, and $s$ and $t$ become equivalent.% Hence $s$
% (from rule one). 
The proof of the second part is straightforward.
\end{proof}

\noindent 
\begin{appxlem}  (Root Conflict) 
$h(s) \downarrow_{R_1} f(t_1, t_2)$ iff either
\begin{center}$s \rightarrow^! a, t_1 \rightarrow^! a, t_2 \rightarrow^! c$\end{center} \begin{center}or\end{center} \begin{center}$s \rightarrow^! b, t_1 \rightarrow^!~b, t_2 \rightarrow^!~c$.\end{center}
\end{appxlem}
\noindent
\begin{proof} The \emph{if} part is straightforward. 
If $s$ and $t_1$ reduce to $a$ (resp., $b$) and $t_2$ reduces to $c$,
then $h(a)$ reduces to $f(a, c)$ (resp., $f(b, c)$).\\

$\emph{Only\;if}$ part: Suppose $h(s)$ is joinable with $f(t_1,\;t_2)$
modulo $R_1$.  We can assume wlog that $s,\;t_1,\;t_2$ are in normal
forms. Then $h(s)$ must be reducible, i.e., $s = a$ or $s = b$. If $s
= a$, then $t_1=a$ and $t_2=c$; else if $s = b$, then $t_1=b$ and
$t_2=c$ (from our rules). 
\end{proof}

Now for $E$-unification, we have the inference rules\\[+10pt]

\begin{tabular}{lcc}
(a) & & $\vcenter{
\infer[\qquad \mathrm{if} ~ X ~ \mathrm{occurs ~ in} ~ \eq \;\; or \;\; \Gamma]{\{X \approx_{}^? V\} \; \cup \; [V/X](\eq) \; \parallel \; [V/X](\Gamma) }
      { \{X \approx_{}^? V\} ~ \uplus ~ \eq \; \parallel \; \Gamma}
}
$\\[+30pt]
(b) & & $\vcenter{
\infer{\eq ~ \cup ~ \{ X \approx_{}^? h(Y), \; T \approx_{}^? Y \} \; \parallel \; \Gamma}
{\eq ~ \uplus ~ \{ X \approx_{}^? h(Y), \; X \approx_{}^? h(T) \} \; \parallel \; \Gamma}
}
$\\[+30pt]

(c) & & $\vcenter{
\infer{\eq ~ \cup ~ \{ X \approx_{}^? f(V, Y), \; W \approx_{}^? V, \; T \approx_{}^? Y \} \; \parallel \; \Gamma}
{\eq ~ \uplus ~ \{ X \approx_{}^? f(V, Y), \; X \approx_{}^? f(W, T) \} \; \parallel \; \Gamma}
}
$\\[+30pt]
(d) & & $\vcenter{
\infer{\eq ~ \cup ~ \{ U \approx_{}^? Y , \; V \approx_{}^? c , \; X \approx_{}^? f(Y, V) \} \; \parallel \; \Gamma ~ \cup ~ \{ (Y = a) \; \vee \; (Y = b) \} }
{\eq ~ \uplus ~ \{ X \approx_{}^? h(Y), \; X \approx_{}^? f(U, V) \} \; \parallel \; \Gamma}
}
$\\[+30pt]
\end{tabular}

The above inference rules are applied with rule~(a) having the highest
priority and rule~(d) the lowest. 

The following are the failure rules, which, of course, have the highest priority.\\

%%%%Failure rules%%%%
\begin{tabular}{lcl}

$(F1)$ & & $\vcenter{
\infer[\qquad d \in \{a,b,c\}]{FAIL}
{\eq ~ \uplus ~ \{ X \approx_{}^? d, \; X \approx_{}^? f(U, V) \} \; \parallel \; \Gamma }
}
$\\[+30pt]
$(F2)$ & & $\vcenter{
\infer[\qquad \; \; \; \, d \in \{a,b,c\}]{FAIL}
{\eq ~ \uplus ~\{ X \approx_{}^? d, \; X \approx_{}^? h(V) \} \; \parallel \; \Gamma}
}
$\\[+30pt]
$(F3)$ & & $\vcenter{
\infer[\qquad \qquad \, d \in \{a,b\}]{FAIL}
{\eq ~ \uplus ~\{ X \approx_{}^? c, \; X \approx_{}^? d \} \; \parallel \; \Gamma}
}
$\\[+30pt]
$(F4)$ & & $\vcenter{
\infer{FAIL }
{\eq ~ \uplus ~\{ X \approx_{}^? b, \; X \approx_{}^? a \} \; \parallel \; \Gamma}
}
$\\[+30pt]
\end{tabular}

\noindent 
\begin{appxlem} 
$R_1^{}$ is \emph{non-subterm-collapsing,} i.e., no 
term is equivalent to a proper subterm of it.
\end{appxlem}
\begin{proof} Since the rules in $R_1$ are size increasing, no term can be reduced to a proper subterm of it.
\end{proof}

Because of the above lemma, we can have an extended occur-check or cycle
check~\cite{JouannaudKirchner91} as another failure rule.\\[-10pt]

\begin{tabular}{lcc}
(F5)& & $\vcenter{
\infer{FAIL}%[\qquad X_i^{} \in V, ~ \; s_j^{} \not\in V ]
      { \{X_0^{} \approx_{}^? s_1^{}[X_1^{}], \; \ldots , \; X_n^{} \approx_{}^? s_n^{}[X_0^{}] \} ~ \uplus ~ \eq \; \parallel \; \Gamma}
}
$\\[+10pt]
\end{tabular}

\noindent
where the $X_i^{}$'s are variables and $s_j^{}$'s are non-variable terms.

Once these inference rules have been exhaustively applied, we are left with
a set of equations in \emph{dag-solved form} along with clausal
constraints. Thus the set of equations is of the form \[ \left\{ \vphantom{b^b}
X_1^{} =^? t_1^{} , \; \ldots \, ~ , \; X_m^{} =^? t_m^{} \right\} \] where
the variables on the left-hand sides are all distinct (i.e.,
$X_i^{} \neq X_j^{}$ for~$i \neq j$). The clausal constraints are 
either negative unit clauses of the form~$\neg(Y = a)$ or~$\neg(Y = b)$
or positive two-literal clauses of the form~$(Y = a) \; \vee \; (Y = b)$.
The solvability of such a system of equations and clauses
can be checked in polynomial time.\\

\pagebreak

\noindent
$\underline{\emph{Steps\;for\;polynomial\;time\;solvability\; of\;equations\;and\;clauses}}$: 
\begin{enumerate}
%\item Check for solvability of equations $\eq$ using known, linear, standard unification algorithms~\cite{DECHAMPEAUX198679}.
\item Add to the list of clauses $\Gamma$ more clauses 
derived from the solved form, to generate
$\, \Gamma_{}^{\prime} \,$. For example if we have an equation of the
form $X \approx^? h(Y)$, then $X \not= a$ and $X \not= b$ will be
added to $\, \Gamma_{}^{\prime} \,$.
\item Check for satisfiability of $\, \Gamma_{}^{\prime}$ by unit resolution
with the negative clauses.
\end{enumerate}

\noindent
\textbf{Soundness} of this algorithm follows from the lemmas
\emph{\textbf{A.1}} through \emph{\textbf{A.4.}}\\

As for
\textbf{termination}, we first observe that none of the inference rules
introduce a new variable, i.e., the number of variables never increases.
With the first inference rule which removes
asymmetry, asymmetric equations are eliminated from $\eq$, i.e., the number
of asymmetric equations goes down. For the
$E$-unification rules, we can see
that in each case either the overall size of equations decreases or some
function
symbols are lost.  
In rule $(a)$, we replace $X$ by $V$ and are left with
an isolated~$X$, hence the number of unsolved variables go down~\cite{Term}.
In rules $(b) \text{ and } (d)$ the number of occurrences of~$h$ goes
down and in rule~$(c)$ the number of occurrences of~$f$ goes down.
%This way the size is decreasing or remains constant.
\\

\ignore{
(h) & & $\vcenter{
\infer[\qquad \mathrm{if} ~ Y \, \succ_h^+ \, X]{\{ X \approx_{}^? h(X) \} \; \cup \; [X/Y](\eq) \}}
{\eq ~ \uplus ~ \{ X \approx_{}^? h(Y) \}}
}
$\\[+30pt]
}

\label{appendixAend}

\pagebreak

%\appendix

%\renewcommand\thesection{\Alph{section}}
\section{Asymmetric Unifiability modulo $R_5^{}$}
\label{appendix2}

\noindent
Recall that the term rewriting system~$R_5^{}$ is
\begin{eqnarray*}
g(a) & \rightarrow & f(a,a,a)\\
g(b) & \rightarrow & f(b,b,b)
\end{eqnarray*}

We assume that the input equations are in standard form, i.e.,
of one of four kinds: $X \approx_{}^? Y$, $X \approx_{}^? g(Y)$,
$X \approx_{}^? f(U,V,W)$ and $X \approx_{}^? d$ where $X, Y, U, V, W$~are variables
and $d$~is any constant. Asymmetric equations will have the extra downarrow,
e.g., $X \approx_{\downarrow}^{?} g(Y)$.\\

As in Appendix~\emph{A}, our algorithm transforms an asymmetric unification problem to a
set of equations in {dag-solved form}
along with {clausal constraints,} where each atom 
is of the form \mbox{$( \langle variable \rangle \; = \; \langle constant \rangle )$}. 
We use the notation
$EQ \; \parallel \; \Gamma$, where $EQ$ is set of equations in
standard form as mentioned above, and $\Gamma$~is a set of clausal
constraints. Initially $\Gamma$ is empty.\\

We first apply the following inference rule (until finished)
that gets rid of asymmetry:\\

\begin{tabular}{lcc}
 &  & $\vcenter{
\infer{\eq ~ \uplus ~ \{X \approx_{}^? g(Y)\} \; \parallel \; \Gamma ~ \cup ~ \{
\neg(Y = a) \} ~ \cup ~ \{\neg(Y = b) \}  }
      { \eq ~ \uplus ~ \{X \approx_{\downarrow}^? g(Y)\} \; \parallel \; \Gamma}
}
$\\[+30pt]
\end{tabular}

Now for $E$-unification, we have the inference rules\\[+10pt]

\begin{tabular}{lcc}
(a) & & $\vcenter{
\infer[\qquad \mathrm{if} ~ X ~ \mathrm{occurs ~ in} ~ \eq ]{\{X \approx_{}^? V\} \; \cup \; [V/X](\eq) \; \parallel \; [V/X](\Gamma) }
      { \{X \approx_{}^? V\} ~ \uplus ~ \eq \; \parallel \; \Gamma}
}
$\\[+30pt]
(b) & & $\vcenter{
\infer{\eq ~ \cup ~ \{ X \approx_{}^? g(Y), \; T \approx_{}^? Y \} \; \parallel \; \Gamma}
{\eq ~ \uplus ~ \{ X \approx_{}^? g(Y), \; X \approx_{}^? g(T) \} \; \parallel \; \Gamma}
}
$\\[+30pt]
(c) & & $\vcenter{
\infer{\eq ~ \cup ~ \{ X \approx_{}^? f(U_1, V_1, W_1), \; U_1 \approx_{}^? U_2, \; V_1 \approx_{}^? V_2, \; W_1 \approx_{}^? W_2 \} \; \parallel \; \Gamma}
{\eq ~ \uplus ~ \{ X \approx_{}^? f(U_1, V_1, W_1), \; X \approx_{}^? f(U_2, V_2, W_2) \} \; \parallel \; \Gamma}
}
$\\[+30pt]
(d) & & $\vcenter{
\infer{\eq ~ \cup ~ \{ U \approx_{}^? Y , \; V \approx_{}^? Y ,\; W \approx_{}^? Y , \; X \approx_{}^? f(Y, Y, Y) \} \; \parallel \; \Gamma ~ \cup ~ \{ (Y = a) \; \vee \; (Y = b) \} }
{\eq ~ \uplus ~ \{ X \approx_{}^? g(Y), \; X \approx_{}^? f(U, V, W) \} \; \parallel \; \Gamma}
}
$\\[+30pt]
\end{tabular}

The inference rules are applied in the descending order of priority
from~(a), the highest, to ~(d) the lowest. Occurrence of equations of
the form $X\approx_{}^?a$ and $X\approx_{}^?f(U, V, W)$ will make 
the equations unsolvable. Hence we have failure rules as in Appendix~\emph{A}.
Since the equational theory is non-subterm-collapsing, we have an
extended occur-check or cycle check rule here as well:\\[-7pt]

\begin{tabular}{lcc}
(Cycle-check)& & $\vcenter{
\infer{FAIL}%[\qquad X_i^{} \in V, ~ \; s_j^{} \not\in V ]
      { \{X_0^{} \approx_{}^? s_1^{}[X_1^{}], \; \ldots , \; X_n^{} \approx_{}^? s_n^{}[X_0^{}] \} ~ \uplus ~ \eq \; \; \parallel \; \; \Gamma}
}
$\\[+10pt]
\end{tabular}

\noindent
where the $X_i^{}$'s are variables and $s_j^{}$'s are non-variable terms.

After exhaustively applying these inference rules we are 
left with
a set of equations in \emph{dag-solved form} along with clausal
constraints. Recall that the clausal constraints are 
either unit clauses of the form~$\neg(W = a)$ or~$\neg(W = b)$
or positive two-literal clauses of the form~$(W = a) \; \vee \; (W = b)$.
The solvability of such a system of equations and clauses
can be checked in polynomial time as described in Appendix~$A$.\\

Similarly, soundness and termination can be shown
as is done in Appendix~$A$.

\ignore{
(h) & & $\vcenter{
\infer[\qquad \mathrm{if} ~ Y \, \succ_h^+ \, X]{\{ X \approx_{}^? h(X) \} \; \cup \; [X/Y](\eq) \}}
{\eq ~ \uplus ~ \{ X \approx_{}^? h(Y) \}}
}
$\\[+30pt]
}

\label{appendixBend}

\pagebreak
\section{Automata Constructions}

\label{appendix3}

We illustrate how automata are constructed for each equation in standard
form. In order to avoid cluttering up the diagrams
the dead state has been included only for the first automaton.
The missing transitions lead to the dead state by default
for the others. Recall that we are considering the case of one constant~$a$. 
The homomorphism~$\mathsf{h}$ is treated as successor function.\\
\subsection{$\mathsf{P=Q+R}$}
\begin{tikzpicture}[shorten >=1pt,node distance=4cm,on grid,auto]
   \node[state,initial,accepting] (q_0) {$q_0$};
   \node[state] (D) [below=of q_0] {$D$};
   \path[->]
    (q_0) edge [loop above] node 
    {$ \begin{psmallmatrix}
           0 \\
           0 \\
           0 
         \end{psmallmatrix}, \begin{psmallmatrix}
           0 \\
           1 \\
           1 
          \end{psmallmatrix}, \begin{psmallmatrix}
           1 \\
           0 \\
           1 
          \end{psmallmatrix}, \begin{psmallmatrix}
           1 \\
           1 \\
           0 
          \end{psmallmatrix}$} ()
             edge node {$ \begin{psmallmatrix}
           0 \\
           0 \\
           1 
          \end{psmallmatrix}, \begin{psmallmatrix}
           0 \\
           1 \\
           0 
          \end{psmallmatrix}, \begin{psmallmatrix}
           1 \\
           0 \\
           0 
          \end{psmallmatrix}, \begin{psmallmatrix}
           1 \\
           1 \\
           1
         \end{psmallmatrix}$} (D)
       (D) edge [loop below] node{$\begin{psmallmatrix}
           0 \\
           0 \\
           0 
          \end{psmallmatrix},\begin{psmallmatrix}
           0 \\
           0 \\
           1 
          \end{psmallmatrix},\begin{psmallmatrix}
           0 \\
           1 \\
           0 
          \end{psmallmatrix},\begin{psmallmatrix}
           0 \\
           1 \\
           1 
          \end{psmallmatrix},\begin{psmallmatrix}
           1 \\
           0 \\
           0 
          \end{psmallmatrix},\begin{psmallmatrix}
           1 \\
           0 \\
           1 
          \end{psmallmatrix},\begin{psmallmatrix}
           1 \\
           1 \\
           0 
          \end{psmallmatrix},\begin{psmallmatrix}
           0 \\
           1 \\
           1 
          \end{psmallmatrix}$} ();  
\end{tikzpicture}

\noindent
Let $\mathsf{P_i, Q_i} \text{ and } \mathsf{R_i}$ denote 
the ${i^{th}}$ bits of $\mathsf{P, Q} \text{ and } \mathsf{R} \;respectively$. 
$\mathsf{P_i}$ has a value~1, when either
$\mathsf{Q_i}$ or $\mathsf{R_i}$ has a value 1. We need 3-bit alphabet
symbols for this equation. For example, if
$\mathsf{R_2}$ = 0, $\mathsf{Q_2}$ = 1, then $\mathsf{P_2}$ = 1. The corresponding
alphabet symbol is $\begin{psmallmatrix}
           P_2 \\
           Q_2 \\
           R_2
          \end{psmallmatrix}$ = $\begin{psmallmatrix} 1 \\ 0 \\ 1
          \end{psmallmatrix}$. \\
          Hence, only strings with the alphabet symbols $\{$ $\begin{psmallmatrix}
           0 \\
           0 \\
           0 
         \end{psmallmatrix}, \begin{psmallmatrix}
           0 \\
           1 \\
           1 
          \end{psmallmatrix}, \begin{psmallmatrix}
           1 \\
           0 \\
           1 
          \end{psmallmatrix}, \begin{psmallmatrix}
           1 \\
           1 \\
           0 
          \end{psmallmatrix}$ $\}$ are accepted by this automaton. 
Rest of the input symbols like $\{$ $\begin{psmallmatrix}
           0 \\
           0 \\
           1 
         \end{psmallmatrix}, \begin{psmallmatrix}
           1 \\
           1 \\
           1 
          \end{psmallmatrix}, \begin{psmallmatrix}
           0 \\
           1 \\
           0 
          \end{psmallmatrix}, \begin{psmallmatrix}
           1 \\
           0 \\
           0 
          \end{psmallmatrix}$ $\}$ go to the dead state~$D$ as they violate 
the XOR property. 

Note that the string $\begin{psmallmatrix}
           1 \\
           0 \\
           1 
          \end{psmallmatrix} \begin{psmallmatrix}
           1 \\
           1 \\
           0 
          \end{psmallmatrix}$ is accepted by automaton. This corresponds
to $\mathsf{P=a+h(a)}$. 
$\mathsf{Q=h(a)}$ and~$\mathsf{R=a}$.\\
\subsection{{$\mathsf{P=_{\downarrow}Q+R}$}}
\begin{tikzpicture}[shorten >=1pt,node distance=7cm,on grid,auto]
   \node[state,initial] (q_0) {$q_0$};  
   \node[state] (q_1) [right=of q_0] {$q_1$};
   \node[state] (q_3) [below=of q_0] {$q_3$};
  % \node[state] (D) [above left=of q_3] {$D$};
   \node[state,accepting] (q_2) [below=of q_1] {$q_2$};
   \path[->]
   (q_0) edge [loop above] node 
    {$ \begin{psmallmatrix}
           0 \\
           0 \\
           0 
         \end{psmallmatrix}$} ()
       %  edge node {$ \begin{psmallmatrix}
          % 0 \\
           %0 \\
           %1 
          %\end{psmallmatrix}, \begin{psmallmatrix}
           %0 \\
           %1 \\
           %0 
          %\end{psmallmatrix}, \begin{psmallmatrix}
           %0 \\
          % 1 \\
           %1 
          %\end{psmallmatrix}, \begin{psmallmatrix}
        %   1 \\
        %   1 \\
        %1
        % \end{psmallmatrix}$} (D)
         edge node {$ \begin{psmallmatrix}
           1 \\
           0 \\
           1 
          \end{psmallmatrix}$} (q_3)
          edge node {$ \begin{psmallmatrix}
           1 \\
           1 \\
           0 
          \end{psmallmatrix}$} (q_1)
          (q_1) edge [loop above] node 
    {$ \begin{psmallmatrix}
           0 \\
           0 \\
           0 
         \end{psmallmatrix},\begin{psmallmatrix}
           1 \\
           1 \\
           0 
         \end{psmallmatrix}$} ()
         edge node {$ \begin{psmallmatrix}
           1 \\
           0 \\
           1 
          \end{psmallmatrix}$} (q_2)
          (q_2) edge [loop right] node 
    {$ \begin{psmallmatrix}
           0 \\
           0 \\
           0 
         \end{psmallmatrix}$,$ \begin{psmallmatrix}
           1 \\
           0 \\
           1 
         \end{psmallmatrix}$,$ \begin{psmallmatrix}
           1 \\
           1 \\
           0 
         \end{psmallmatrix}$} ()
         (q_3) edge [loop left] node 
    {$ \begin{psmallmatrix}
           0 \\
           0 \\
           0 
         \end{psmallmatrix}$,$ \begin{psmallmatrix}
           1 \\
           0 \\
           1 
         \end{psmallmatrix}$} ()
         edge node 
    {$ \begin{psmallmatrix}
           1 \\
           1 \\
           0 
         \end{psmallmatrix}$} (q_2);
 \end{tikzpicture}
 
\noindent
To preserve asymmetry on the right-hand side of this equation, $\mathsf{Q+R}$ should be 
irreducible. If either $\mathsf{Q}$ or $\mathsf{R}$ is empty, or if they 
have any term in common, then a reduction will occur. For example, 
if $\mathsf{Q}$ = $\mathsf{h(a)}$ and $\mathsf{R}$ = $\mathsf{h(a)+a}$, 
there is a reduction, whereas if $\mathsf{R}$ = $\mathsf{h(a)}$ 
and $\mathsf{Q}$ = $\mathsf{a}$, irreducibility is preserved, since 
there is no common term and neither one is empty. 
Since neither $\mathsf{Q}$ nor $\mathsf{R}$ can be empty,
any accepted string should have one occurrence of $\begin{psmallmatrix}
           1 \\
           0 \\
           1 
         \end{psmallmatrix}$ and one occurrence of~$\begin{psmallmatrix}
           1 \\
           1 \\
           0 
         \end{psmallmatrix}$.
\pagebreak
%%%%X=h(Y)%%%%
\subsection{$\mathsf{X=h(Y)}$}
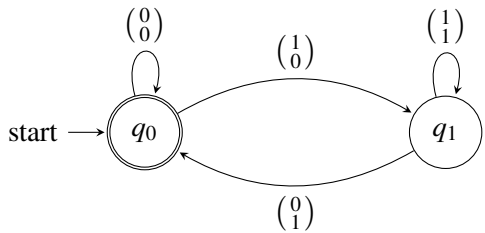
\begin{tikzpicture}[>=stealth,shorten >=1pt,auto,node distance=4cm]
    \node[state,initial,accepting] (q_0)                {$q_0$};
    \node[state] (q_1) [right of = q_0] {$q_1$};
    \path[->] (q_0) edge [bend left]  node {$\begin{psmallmatrix}
           1 \\
           0 
          \end{psmallmatrix}$} (q_1)
          edge [loop above] node {$\begin{psmallmatrix}
           0 \\
           0 
          \end{psmallmatrix}$} ()
              (q_1) edge [bend left] node {$\begin{psmallmatrix}
           0 \\
           1 
          \end{psmallmatrix}$} (q_0)
          edge [loop above] node {$\begin{psmallmatrix}
           1 \\
           1 
          \end{psmallmatrix}$} ();
\end{tikzpicture}

\noindent
We need 2-bit vectors as alphabet symbols since we have 
two unknowns $\mathsf{X}$ and $\mathsf{Y}$.
Note again that $\mathsf{h}$ acts like the successor function.
$q_0$ is the only accepting state. 
A state transition occurs with bit vectors $\begin{psmallmatrix}
           1 \\
           0         \end{psmallmatrix},\begin{psmallmatrix}
           0 \\
           1         \end{psmallmatrix}$. If $\mathsf{Y}$=1 in current state, then $\mathsf{X}$=1 in the next state, hence a transition occurs from $\mathsf{q_0}$ to $\mathsf{q_1}$, and vice versa. The ordering of variables is $\begin{psmallmatrix}
           Y \\
           X         \end{psmallmatrix}$.
           % Transitions to dead state $D$ have been excluded from this automaton.\\

%%%%X=(asymmetric)h(Y)%%%%
\subsection{$\mathsf{X=_{\downarrow}h(Y)}$}
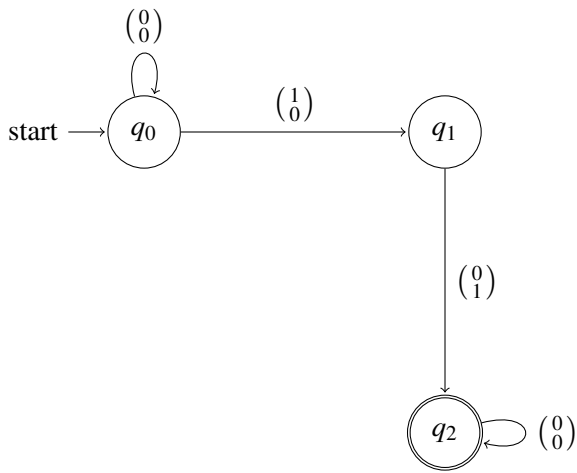
\begin{tikzpicture}[shorten >=1pt,node distance=4cm,on grid,auto]
   \node[state,initial] (q_0) {$q_0$};
 %  \node[state] (D) [below=of q_0] {$D$};
   \node[state] (q_1) [right=of q_0] {$q_1$};
   \node[state,accepting] (q_2) [below=of q_1] {$q_2$};
   \path[->]
    (q_0) edge node {$\begin{psmallmatrix}
           1 \\
           0 
          \end{psmallmatrix}$} (q_1)
          %edge node {$\begin{psmallmatrix}
           %1 \\
           %1 
          %\end{psmallmatrix}
          %,\begin{psmallmatrix}
           %0 \\
           %1 
          %\end{psmallmatrix}$} (D)
          edge [loop above] node 
    {$ \begin{psmallmatrix}
           0 \\
           0 
         \end{psmallmatrix}$} ()
        (q_1) edge node {$\begin{psmallmatrix}
           0 \\
           1 
          \end{psmallmatrix}$} (q_2)
          (q_2) edge [loop right] node{$\begin{psmallmatrix}
           0 \\
           0
          \end{psmallmatrix}$}();
          \end{tikzpicture}

\noindent
In 
this equation, $\mathsf{h(Y)}$ should be in normal form. So
$\mathsf{Y}$ cannot be 
either 0 or of the form $\mathsf{u + v}$. Thus
$\mathsf{Y}$ has to be a string of the form $0_{}^i 1 0_{}^j$
and $\mathsf{X}$ then has to be $0_{}^{i+1} 1 0_{}^{j-1}$.
Therefore
the bit vector $\begin{psmallmatrix}
           1 \\
           0         \end{psmallmatrix}$ has to be succeeded by $\begin{psmallmatrix}
           0 \\
           1         \end{psmallmatrix}$.\\ 
\pagebreak
\subsection{An Example}

Let $\left\{ \vphantom{b^b} 
{U=_{\downarrow}V+Y}, \; 
{W=h(V)}, \;
{Y=_{\downarrow}h(W)} \right\}$ be an asymmetric unification
problem.
We need 4-bit vectors and 3~automata since we have 4 unknowns in 3 equations, with bit-vectors represented in this ordering of set variables: $ \begin{psmallmatrix}
           V \\
            W\\
           Y \\
           U
         \end{psmallmatrix}$. \\\\\\\\

\noindent
$\mathsf{\bf{Y=_{\downarrow}h(W)}}$\\\\

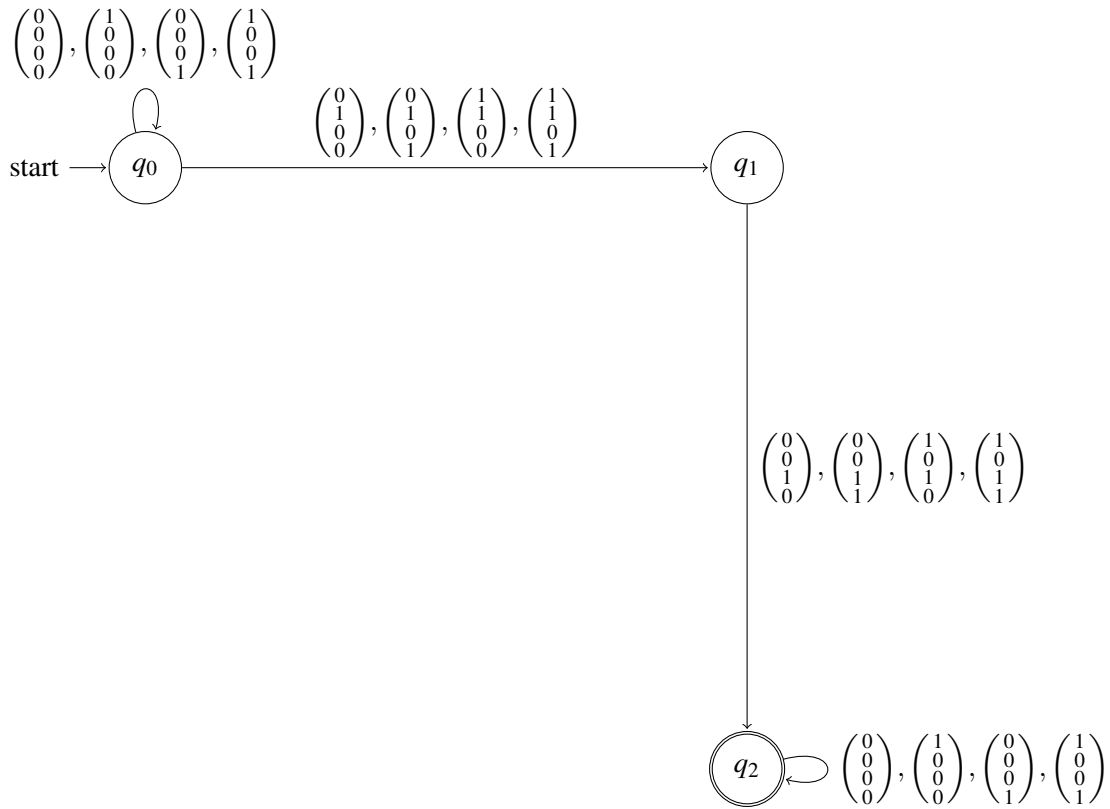
\begin{tikzpicture}[shorten >=1pt,node distance=8cm,on grid,auto]
   \node[state,initial] (q_0) {$q_0$};
   \node[state] (q_1) [right=of q_0] {$q_1$};
   \node[state,accepting] (q_2) [below=of q_1] {$q_2$};
   \path[->]
    (q_0) edge node {$\begin{psmallmatrix}
           0 \\
           1 \\
           0\\
           0
          \end{psmallmatrix},\begin{psmallmatrix}
           0 \\
           1 \\
           0\\
           1
          \end{psmallmatrix},\begin{psmallmatrix}
           1 \\
           1 \\
           0\\
           0
          \end{psmallmatrix},\begin{psmallmatrix}
           1 \\
           1 \\
           0\\
           1
          \end{psmallmatrix}$} (q_1)
          edge [loop above] node 
    {$ \begin{psmallmatrix}
           0 \\
           0 \\
           0\\
           0
         \end{psmallmatrix},\begin{psmallmatrix}
           1 \\
           0 \\
           0\\
           0
         \end{psmallmatrix},\begin{psmallmatrix}
           0 \\
           0 \\
           0\\
           1
         \end{psmallmatrix},\begin{psmallmatrix}
           1 \\
           0 \\
           0\\
           1
         \end{psmallmatrix}$} ()
        (q_1) edge node {$\begin{psmallmatrix}
           0 \\
           0\\
           1\\
           0 
          \end{psmallmatrix},\begin{psmallmatrix}
           0 \\
           0\\
           1\\
           1 
          \end{psmallmatrix},\begin{psmallmatrix}
           1 \\
           0\\
           1\\
           0 
          \end{psmallmatrix},\begin{psmallmatrix}
           1 \\
           0\\
           1\\
           1 
          \end{psmallmatrix}$} (q_2)
         (q_2) edge [loop right] node 
    {$ \begin{psmallmatrix}
           0 \\
           0 \\
           0\\
           0
         \end{psmallmatrix},\begin{psmallmatrix}
           1 \\
           0 \\
           0\\
           0
         \end{psmallmatrix},\begin{psmallmatrix}
           0 \\
           0 \\
           0\\
           1
         \end{psmallmatrix},\begin{psmallmatrix}
           1 \\
           0 \\
           0\\
           1
         \end{psmallmatrix}$} ();
          \end{tikzpicture}
 %\noindent
 
\pagebreak

\noindent
$\mathsf{\bf{U=_{\downarrow}V+Y}}$\\

\noindent
We include the $\times$ (``don't-care'') symbol in state transitions to
indicate that the values can be either $0$ or~$1$. This is essentially
to avoid
cluttering the diagrams. Note that here
this $\times$~symbol is
a placeholder for the variable~$W$ which does not have any
significance in this automaton.

\noindent
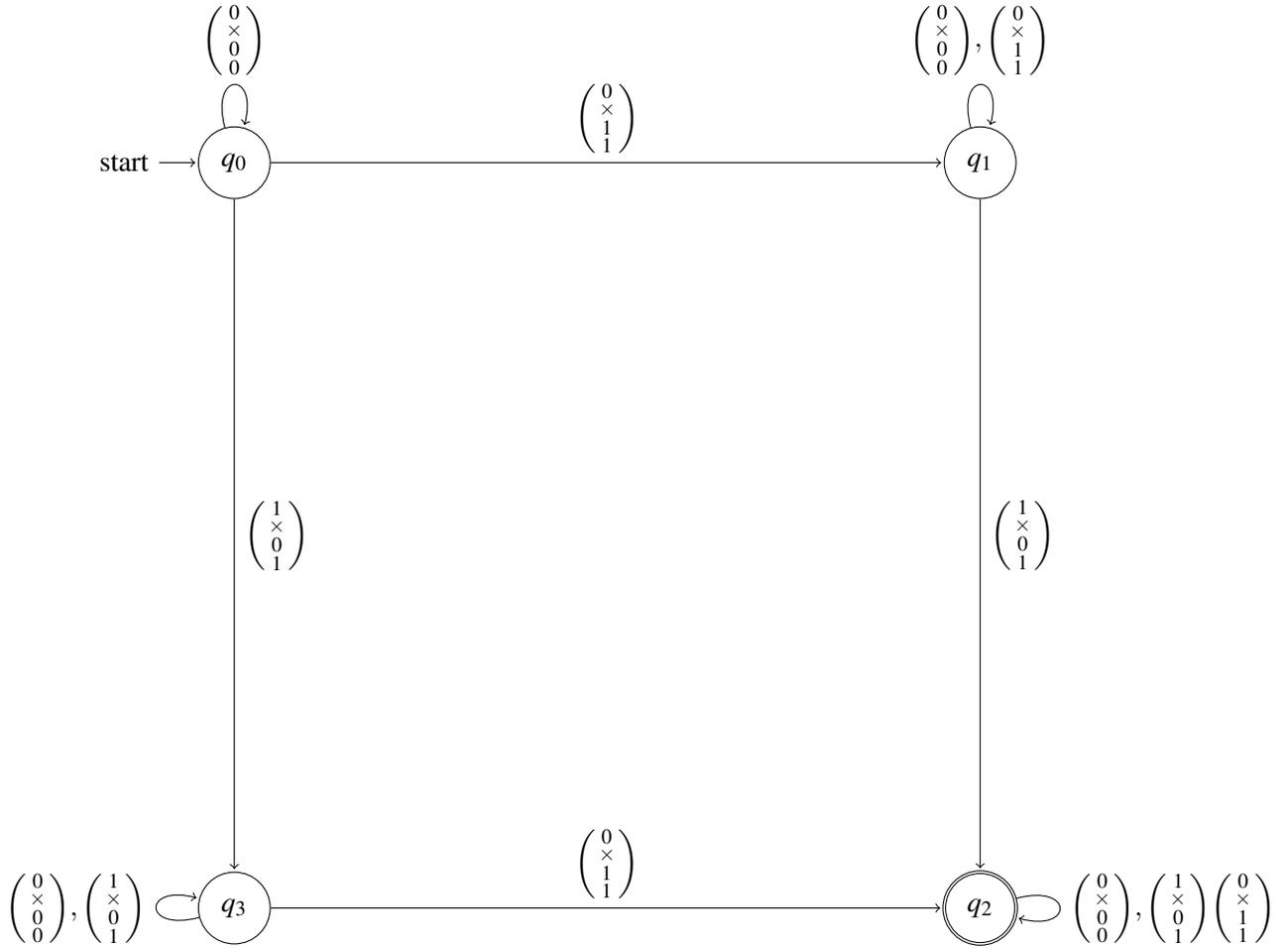
\begin{tikzpicture}[shorten >=1pt,node distance=10cm,on grid,auto]
   \node[state,initial] (q_0) {$q_0$};  
   \node[state] (q_1) [right=of q_0] {$q_1$};
   \node[state] (q_3) [below=of q_0] {$q_3$};
   \node[state,accepting] (q_2) [below=of q_1] {$q_2$};
   \path[->]
   (q_0) edge [loop above] node 
    {$ \begin{psmallmatrix}
           0 \\
           \times \\
           0 \\
           0
         \end{psmallmatrix}$} ()
         edge node {$ \begin{psmallmatrix}
           1 \\
           \times \\
           0 \\
           1 
          \end{psmallmatrix}$} (q_3)
           edge node {$ \begin{psmallmatrix}
           0 \\
           \times\\
           1 \\
           1 
          \end{psmallmatrix}$} (q_1)
          (q_1) edge [loop above] node 
    {$ \begin{psmallmatrix}
           0 \\
           \times \\
           0 \\
           0
         \end{psmallmatrix},\begin{psmallmatrix}
           0 \\
           \times \\
           1 \\
           1
         \end{psmallmatrix}$} ()
         edge node {$ \begin{psmallmatrix}
           1 \\
           \times \\
           0 \\
           1 
          \end{psmallmatrix}$} (q_2)
          (q_2) edge [loop right] node 
    {$ \begin{psmallmatrix}
           0 \\
           \times \\
           0 \\
           0
         \end{psmallmatrix},\begin{psmallmatrix}
           1 \\
           \times \\
           0 \\
           1
         \end{psmallmatrix}\begin{psmallmatrix}
           0 \\
           \times \\
           1 \\
           1
         \end{psmallmatrix}$} ()
         (q_3) edge [loop left] node 
    {$ \begin{psmallmatrix}
           0 \\
           \times \\
           0 \\
           0
         \end{psmallmatrix},\begin{psmallmatrix}
           1 \\
           \times \\
           0 \\
           1
         \end{psmallmatrix}$} ()
         edge node 
    {$\begin{psmallmatrix}
           0 \\
           \times\\
           1 \\
           1 
          \end{psmallmatrix}$} (q_2);
 \end{tikzpicture}
 
\pagebreak

\noindent
$\mathsf{\bf{W=h(V)}}$\\

 
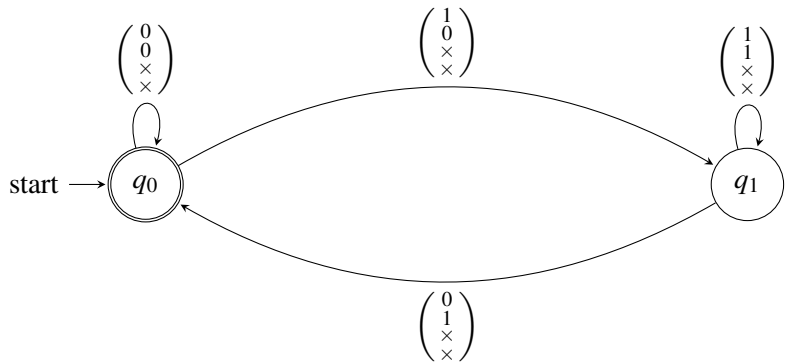
\begin{tikzpicture}[>=stealth,shorten >=1pt,auto,node distance=8cm]
    \node[state,initial,accepting] (q_0)                {$q_0$};
    \node[state] (q_1) [right of = q_0] {$q_1$};
    \path[->] (q_0) edge [bend left]  node {$\begin{psmallmatrix}
           1 \\
           0 \\
           \times\\
            \times
          \end{psmallmatrix}$} (q_1)
          edge [loop above] node {$\begin{psmallmatrix}
           0 \\
           0 \\
           \times \\
            \times
          \end{psmallmatrix}$} ()
              (q_1) edge [bend left] node {$\begin{psmallmatrix}
              0\\
           1 \\
            \times \\
           \times
          \end{psmallmatrix}$} (q_0)
          edge [loop above] node {$\begin{psmallmatrix}
           1 \\
           1 \\
            \times\\
            \times
          \end{psmallmatrix}$} ();
\end{tikzpicture}\\
\textbf{NOTE}: As before, the symbol $\times$ in the vectors means that 
the bit value can be either~0 or~1. \\

\addtolength{\baselineskip}{12pt}
The string~$\begin{psmallmatrix}
              1\\
           0 \\
            0 \\
           1
          \end{psmallmatrix}$$\begin{psmallmatrix}
              0\\
           1 \\
            0 \\
           0
          \end{psmallmatrix}$$\begin{psmallmatrix}
              0\\
           0 \\
            1 \\
           1
          \end{psmallmatrix}$ $\begin{psmallmatrix}
              0\\
           0 \\
            0 \\
           0
          \end{psmallmatrix}$
is accepted by all the three automata.
The corresponding
asymmetric unifier is \[ \mathsf{ \left\{V \mapsto a, \, W\mapsto h(a), \, 
Y\mapsto h^2(a), \, U \mapsto ( h^2(a)+a ) \right\}}. \]

\label{appendix4}
\pagebreak

\nocite{DBLP:journals/jacm/BuntineB94,DBLP:journals/jsc/ComonL89}
\end{document}